\documentclass[11pt,a4paper]{article}
\usepackage[T1]{fontenc}
\usepackage[utf8]{inputenc}
\usepackage{authblk}
\usepackage{amscd}
\usepackage{amsmath}
\usepackage{mathrsfs}
\usepackage{cases}
\usepackage{amssymb}
\usepackage{amsfonts}
\usepackage{amssymb,amsfonts}

\usepackage{amsmath,amssymb,amsthm,hyperref}

\title{Integrality structures in topological strings and quantum $2$-functions}
\author{Shengmao Zhu \thanks{szhu@zju.edu.cn}}

\newtheorem{thm}{Theorem}[section]
\newtheorem{theorem}[thm]{Theorem}
\newtheorem{lemma}[thm]{Lemma}
\newtheorem{corollary}[thm]{Corollary}
\newtheorem{corollary*}[thm]{Corollary*}

\newtheorem{proposition*}[thm]{Proposition*}
\newtheorem{conjecture}[thm]{Conjecture}
\theoremstyle{definition}

\newtheorem{remark}[thm]{Remark}

\newtheorem{definition}[thm]{Definition}

\newtheorem{defn-thm}[thm]{Definition-Theorem}

\date{} 

\begin{document}
  \maketitle
\begin{abstract}
In this article, we first prove the integrality of open string BPS
numbers for a class of toric Calabi-Yau manifolds named generalized
conifolds, by applying the method introduced in our previous work
\cite{LZ} to the explicit disk counting formula obtained in
\cite{PS}. Then, motivated by the integrality structures in open
topological string theory, we introduce a mathematical notion of
``quantum 2-function'' which can be viewed as the quantization of
the notion of ``2-function'' introduced in \cite{SVW1}. Finally, we
provide a basic example of quantum 2-function and discuss the
quantization of 2-functions.
\end{abstract}

\section{Introduction} \label{section-introduction}
This paper concerns the integrality structures appearing naturally
in topological string theory. The basic example of mirror symmetry
constructed in \cite{COGP} implies the integrality of instanton
numbers $N_{0,d}$ which are defined through the genus zero
Gromov-Witten invariants $K_{0,d}$ of the quintic. More precisely,
the genus $0$ Gromov-Witten potential takes the form
\begin{align} \label{formula-genus0}
F_0=\sum_{d\geq 1}K_{0,d}a^d=\sum_{d\geq 1}N_{0,d}\text{Li}_3(a^d)
\end{align}
where we used the notation of poly-logarithm
$\text{Li}_r(x)=\sum_{k\geq 1}\frac{x^k}{k^r}$, $a$ is a parameter
related to the K\"ahler class of the quintic. Formula
(\ref{formula-genus0}) is usually referred as the multiple covering
formula or Aspinwall-Morrison formula \cite{AM} in literatures.

In general, the Gromov-Witten invariants $K_{0,d}$ are rational
numbers, which is obvious from both the definition in Gromov-Witten
theory, as well as from the B-model computations. However, the
integrality of $N_{0,d}$ is not clear from the formula
(\ref{formula-genus0}). In \cite{KSV}, Kontsevich-Schwarz-Vologodsky
proposed a mathematical proof of the integrality of $N_{0,d}$ by
using the $p$-adic theory, see \cite{SV1,SV2} for further
progresses. The physical explanation of integrality was given in
\cite{GV} by relating $N_{0,d}$ to the degeneracy of BPS states.
More precisely, let $X$ be a Calabi-Yau 3-fold and let $K^{X}_{g,Q}$
be the genus $g$ Gromov-Witten invariant of $X$ in the curve class
$Q\in H_2(X,\mathbb{Z})$, Gopakumar and Vafa \cite{GV} expressed the
Gromov-Witten invariants $K_{g,Q}^X$ in terms of integer invariants
$N_{g,Q}^X$ obtained by BPS state counts
\begin{align} \label{formula-GV}
F^X&=\sum_{g\geq 0}g_s^{2g-2}\sum_{Q\neq 0}K_{g,Q}^Xa^Q\\\nonumber
&=\sum_{g\geq 0, d\geq 1}\sum_{Q\neq
0}\frac{1}{d}N_{g,Q}^X\left(2\sin\frac{dg_s}{2}\right)^{2g-2}a^{dQ}.
\end{align}
Usually, these predicted integer invariants $N_{g,Q}^X$ are referred
as Gopakumar-Vafa invariants in literatures. It is clear that
formula (\ref{formula-genus0}) is the genus $0$ part of the above
formula (\ref{formula-GV}). For a compact Calabi-Yau 3-fold $X$, the
mathematical proof of the integrality of $N_{g,Q}^{X}$ is still
unknown. However, when $X$ is a toric Calabi-Yau 3-fold, the
integrality of  $N_{g,Q}^{X}$ was first proved by P. Peng for the
case of toric Del Pezzo surfaces \cite{Peng}. The proof for general
toric Calabi-Yau 3-folds was then given by Konishi in
\cite{Konishi}. See also \cite{GZ} for several explicit formula of
the Gopakumar-Vafa invariants for local $\mathbb{P}^2$.

Now we consider the open topological strings theory on Calabi-Yau
3-fold $X$. Suppose $L\subset X$ is a Lagrangian submanifold which
may be viewed as the support of a topological D-brane in the
A-model. It is well-known that the classical deformation space of
$L$ modulo Hamiltonian isotopy is unobstructed and of dimension
equal to $b_1(L)$. The superpotential $W$  depending on the K\"ahler
moduli of $X$ and the choice of a flat bundle over $L$, is the
generating function counting worldsheet instanton corrections from
holomophic disks ending on the Lagrangian $L$.

More precisely, the spacetime superpotential can be identified with
the topological disk partition function and is conjectured to admit
an expansion of the general form
\begin{align} \label{formula-disk}
W(a,x)=F^{(X,L)}_{\text{disk}}(a,x)=\sum_{Q,m}K_{0,Q,m}a^Q
x^m=\sum_{Q,m}\sum_{k\geq 1}\frac{n_{0,Q,m}}{k^2}a^{kQ}x^{km}.
\end{align}
where the sum is over relative cohomology classes in $H_2(X,L)$, $a$
denotes the closed string K\"ahler parameters of $X$ and $y$ is the
open string deformation parameters. The final transformation in
(\ref{formula-disk}) is a resummation of the multi-cover
contributions and it is conjectured in \cite{OV} that the resulting
expansion coefficients $n_{0,Q,m}$ are integers which are
interpreted as the counting of BPS states in class $(Q,m)$.


\begin{remark}
Sometimes, such as in \cite{SVW1}, we use $\beta\in H_2(X,L)$ to
denote the class in $H_2(X,L)$, then the above formula
(\ref{formula-disk}) can also be written as
\begin{align} \label{formula-disk2}
W(q)=\sum_{\beta}K_{0,\beta}q^{\beta}=\sum_{\beta}n_{0,\beta}\text{Li}_2(q^{\beta})
\end{align}
where $q$ is the combination of moduli parameter of $(X,L)$.
\end{remark}

  When
$X$ is the quintic and $L$ is the real locus, the superpotential $W$
had been computed in \cite{Wal}. See \cite{AB,SLY,Wal2,Wal3} for
more results about the superpotential $W$ for the compact Calabi-Yau
manifolds. However, the integrality of $n_{0,\beta}$ is not clear
from the formula (\ref{formula-disk2}). A mathematical proof was
proposed in \cite{SV2} follows the work \cite{KSV}.

When $X$ is a toric Calabi-Yau 3-fold and $L$ is the special
Lagrangian submanifold named Aganagic-Vafa A-brane \cite{AV}, the
mirror geometry information of $(X,L)$ is encoded in a mirror curve.
The superpotential ( or the disc counting formula) of $(X,L)$ can be
derived from the mirror curve \cite{AV,AKV}. Moreover, Aganagic and
Vafa surprisingly found the computation by using mirror symmetry and
the result from Chern-Simons knot invariants  are matched. In
\cite{AKV},  Aganagic Klemme and Vafa investigated the integer
ambiguity appearing in the disc counting and discovered that the
corresponding ambiguity in Chern-Simons theory was described by the
framing of the knot. They checked that the two ambiguities match for
the case of the unknot, by comparing the disk amplitudes on both
sides. Motivated by this, one can introduce an integer $\tau$ named
framing to describe the ambiguity. Let $\hat{X}$ to be the resovled
conifold, and $D_\tau$ the Aganagic-Vafa A-brane which is the dual
of the framed unknot $U_\tau$. In \cite{MV}, Mari\~no and Vafa
carefully studied the open topological string on
$(\hat{X},D_{\tau})$, they computed the disk counting amplitude
$F_{\text{disk}}^{(\hat{X},D_\tau)}$ for this model and obtained the
explicit expression for the corresponding integer invariants
$n_{0,Q,m}^{(\hat{X},D_\tau)}$, see also \cite{Zhu4} for this
computations, where we use the notation $n_{m,0,Q}(\tau)$ to denote
this integer invariant instead. The mathematical proof of
integrality $n_{0,Q,m}^{(\hat{X},D_\tau)}$ was given in \cite{LZ}.
Moreover, we find in \cite{LZ0,Zhu3} an interesting explanations of
the integrality of these number by quiver representation theory,
this provides the first example of toric Calabi-Yau and quiver
correspondence, see \cite{Zhu2} for a review of these integrality
results in topological strings. Then in \cite{KRSS1,KRSS2}, a
general knot-quiver correspondence was proposed. This correspondence
for a large class of knot, and links was established
\cite{KRSS2,PSS,EKL,PS}.

Furthermore, with the help of the knot-quiver correspondence, M.
Panfl and P. Sulkowski \cite{PS} obtained an explicit disc counting
formula for the open topological string theory on a class of toric
Calabi-Yau manifolds without compact four-cycles, also referred to
as strip geometries or generalized conifold.

Let $\widehat{X}$  be a generalized conifold with the K\"ahler
parameters arising from two types $a_1,..,a_r$ and $A_1,...,A_s$
where $r,s\geq 0$, and let $D_\tau$  be the framed Aganagic-Vafa
A-brane. Set $\mathbf{l}=(l_1,...,l_r)$, $\mathbf{k}=(k_1,...,k_s)$,
and $|\mathbf{l}|=\sum_{j=1}^rl_j$, $|\mathbf{k}|=\sum_{j=1}^sk_j$.
Given a positive integer $m$, we define
\begin{align} \label{formula-cmlk}
c_{m,\mathbf{l},\mathbf{k}}(\tau)&=\frac{(-1)^{m(\tau+1)+|\mathbf{l}|}}{m^2}
\binom{m\tau+|\mathbf{l}|+|\mathbf{k}|-1}{m-1}\\\nonumber
&\times\prod_{j=1}^{r}\binom{m}{l_j}\prod_{j=1}^s
\frac{m}{m+k_j}\binom{m+k_j}{k_j}.
\end{align}

Then Panfl and Sulkowski obtained the following disk counting
formula for $(\widehat{X},D_{\tau})$ (cf. formula (4.19) in
\cite{PS}):
\begin{align} \label{formula-disk-generalizedconifold}
F_{\text{disk}}^{(\widehat{X},D_\tau)}&=\sum_{m,\mathbf{l},\mathbf{k}}c_{m,\mathbf{l},\mathbf{k}}(\tau)a_1^{l_1}\cdots
a_r^{l_r}A_1^{k_1}\cdots A_s^{k_s}x^m\\\nonumber
&=\sum_{m,\mathbf{l},\mathbf{k}}\sum_{d\geq
1}n_{m,\mathbf{l},\mathbf{k}}(\tau)\frac{1}{d^2}a_1^{dl_1}\cdots
a_r^{dl_r}A_1^{dk_1}\cdots A_s^{dk_s}x^{dm}.
\end{align}

By M\"obius inversion formula, we have the explicit formula for the
disc counting BPS invariants
\begin{align} \label{formula-generaldiscounting}
n_{m,\mathbf{l},\mathbf{k}}(\tau)
=\sum_{d|\text{gcd}(m,\mathbf{l},\mathbf{k})}\frac{\mu(d)}{d^2}c_{m/d,\mathbf{l}/d,\mathbf{k}/d}(\tau)
\end{align}

In this article,  we generalize the method used in \cite{LZ} to
prove that
\begin{theorem} \label{theorem-introduction}
For any $m$, $\mathbf{l}$ and $\mathbf{k}$ given above, we have
\begin{align}
n_{m,\mathbf{l},\mathbf{k}}(\tau)\in \mathbb{Z}.
\end{align}
\end{theorem}

Motivated by the multiple covering formulas (\ref{formula-genus0})
and (\ref{formula-disk}), Schwarz, Vologodsky and Walcher
\cite{SVW1} introduced the mathematical notion of {\em $s$-function}
which is the integral linear combinations of poly-logarithms. We
review the definition and properties of $2$-functions in Section
\ref{Section-quantum2functions}, then it is easy to see that the
proof of Theorem \ref{theorem-introduction}  immediately implies
that
\begin{corollary}
The disk counting formula $F_{\text{disk}}^{(\widehat{X},D_\tau)}$
given by formula (\ref{formula-disk-generalizedconifold}) for the
generalized conifold $(\widehat{X},D_{\tau})$ is a $2$-function.
\end{corollary}

The disc counting formula (\ref{formula-disk}) can be generalized to
the higher genus case. Indeed, based on Ooguri and Vafa's work
\cite{OV}, the generating function of all genus open Gromov-Witten
invariants can also be expressed in terms of a series of new
integers which were later refined by Labastida, Mari\~no and Vafa in
\cite{LM1,LM2,LMV}. Motivated by their results, we formulate a
mathematical notation of quantum $2$-function which can be viewed as
the quantum version of the $2$-function introduced in \cite{SVW1}.
\begin{definition}  We call a  formal
power series
\begin{align} \label{formula-F}
F(\lambda,\mathbf{z},x)=\sum_{g\geq 0,m\geq 1}\sum_{\mathbf{d}>
 \mathbf{0}}\lambda^{2g}K_{g,\mathbf{d},m}\mathbf{z}^{\mathbf{d}}x^m\in
\mathbb{Q}[[\lambda^2,z_1,...,z_r,x]]
\end{align}
with rational coefficients $K_{g,\mathbf{d},m}$ a {\em quantum
$2$-function} if it can be written in the following form
\begin{align} \label{formula-quantum2function}
F(\lambda,\mathbf{z},x)&=\sum_{g\geq 0,m\geq 1}\sum_{\mathbf{d}>
\mathbf{0}}\sum_{k\geq
1}n_{g,\mathbf{d},m}\frac{2\lambda}{km}\sin\left(\frac{km\lambda}{2}\right)
\left(2\sin
\frac{k\lambda}{2}\right)^{2g-2}\mathbf{z}^{\mathbf{kd}}x^{km}
\end{align}
 with $n_{g,\mathbf{d},m}\in \mathbb{Z}$, where we used the multiple-index
 notations
 $\mathbf{z}=(z_1,...,z_r)$, $\mathbf{d}=(d_1,...,d_r)$ and $\mathbf{z}^{\mathbf{d}}=z_1^{d_1}\cdots
 z_r^{d_r}$.
\end{definition}
It is clear that when $\lambda=0$, $F(0,\mathbf{z},x)$ is just the
$2$-function in the sense of \cite{SVW1}. We hope that the quantum
2-function have independent interests in mathematics.

Then we provide a basic example for quantum $2$-function. We
consider the open topological string  model $(\hat{X},D_\tau)$,
where $\hat{X}$ a resolved conifold and $D_\tau$ is the
Aganagic-Vafa A-brane which is the large $N$ duality of the framed
unknot $U_\tau$ with framing $\tau$ in Chern-Simons theory, we
consider the generating function
\begin{align} \label{formula-F-conifold-intro}
F^{(\hat{X},D_\tau)}(\lambda,a,x)=\sum_{g\geq 0,m\geq 1}\sum_{d>
 0}\lambda^{2g}K_{g,d,m}^{(\hat{X},D_\tau)}a^{d}x^m
\end{align}
where $K_{g,d,m}^{(\hat{X},D_\tau)}$ are the   one-hole open
Gromov-Witten invariants of genus $g$ with degree $d$ and writhe
number $m$, whose mathematical definition was given in \cite{KL}. We
will show that the results obtained in our previous work \cite{LZ}
imply that
\begin{theorem}
The function $F^{(\hat{X},D_\tau)}(\lambda,a,x)$ given by formula
(\ref{formula-F-conifold-intro}) is a quantum $2$-function.
\end{theorem}

Finally, we discuss the question how to construct a quantum
2-function by quantizing a $2$-function. Motivated by the method of
topological recursion introduced in \cite{EO1} and its applications
in topological string theory \cite{BKMP,EO2,FLZ}, we briefly
describe a natural method to construct an operator $\mathbf{Q}$ such
that when apply it to a 2-function $W$, then $\mathbf{Q}(W)$ will be
a quantum 2-function.

On the other hand side, Schwarz-Vologodsky-Walcher \cite{SVW1}
introduced a framing transformation operator $\mathbf{f}^{\tau}$
(with $\tau\in \mathbb{Z}$ and $\mathbf{f}^0=id$) on the set of
$2$-functions. They claimed that $\mathbf{f}^{\tau}(W)$ is still a
2-function for any $\tau\in \mathbb{Z}$ if $W$ is a 2-function.
Therefore, we conjecture that $\mathbf{Q}(\mathbf{f}^\tau(W))$ will
be a quantum 2-function for any $\tau\in \mathbb{Z}$ and any
2-function $W$.

\begin{remark} Sometime, it is easy to see that $\mathbf{Q}(W)$ is a
quantum 2-function, but it is very difficult to prove that
$\mathbf{Q}(\mathbf{f}^\tau(W))$ is quantum 2-function for any
$\tau\in \mathbb{Z}$. We leave the further discussions about the
operator $\mathbf{Q}$ and quantum 2-functions to a separated paper.
\end{remark}

\section{Proof of the Theorem \ref{theorem-introduction}}
\label{section-proof} We follow the notations used in \cite{Ko}. Let
$p$ be any prime number, for any nonzero integer $n$, let the
$p$-adic ordinal of $n$, denoted $\text{ord}_{p} n$, be the highest
power of $p$ which divides $n$, i.e. the greatest $\alpha$ such that
$n=p^\alpha m$ for some integer $m$. If $n=0$, we agree to write
$\text{ord}_p 0=\infty$. For any rational number $x=\frac{a}{b}$, we
define
\begin{align}
\text{ord}_{p}x=\text{ord}_pa-\text{ord}_pb.
\end{align}
Given any two rational numbers $x,y\in \mathbb{Q}$, it is obvious
that
\begin{align}
\text{ord}_p(x+y)\geq \min\{\text{ord}_p x,\text{ord}_p y\}.
\end{align}

For nonnegative integer $n$ and prime number $p$, we introduce the
following function
\begin{equation} \label{functionfp}
    f_p(n)=\prod_{i=1,p\nmid i}^n i.
\end{equation}
By its definition, $f_p(n)$ has no $p$-factor, i.e.
$\text{ord}_p(f_p(n))=0$.

\begin{lemma} \label{lemma-1}
Suppose $n\in \mathbb{Z}_+$, for odd prime numbers $p$ and
$\alpha\geq 1$ or for $p=2$, $\alpha\geq 2$,
    we have
\begin{align}
\text{ord}_p(f_p(p^\alpha n)-f_{p}(p^\alpha)^n)\geq 2\alpha.
\end{align}
For $p=2, \alpha=1$,
\begin{align} \label{formula-f2}
\text{ord}_2(f_2(2n)-(-1)^{[n/2]})\geq 2.
\end{align}
\end{lemma}
\begin{proof}
We prove the Lemma \ref{lemma-1} by induction. The case for $n=1$ is
obvious. Now suppose the Lemma \ref{lemma-1}  holds for $n-1$. Since
\begin{align}
&f_p(p^\alpha n)-f_p(p^\alpha)^n)\\\nonumber &=f_p(p^\alpha
n)-f_p(p^\alpha(n-1))f_p(p^\alpha)+f_p(p^\alpha(n-1))f_p(p^\alpha)-f_p(p^\alpha)^n
\end{align}
Then
\begin{align}
&\text{ord}_p(f_p(p^\alpha n)-f_p(p^\alpha)^n)\\\nonumber &\geq
\min\left\{\text{ord}_p(f_p(p^\alpha
n)-f_p(p^\alpha(n-1))f_p(p^\alpha))\right.,\\\nonumber
&\left.\text{ord}_p((f_p(p^\alpha(n-1))-f_p(p^\alpha)^{n-1})f_p(p^\alpha))\right\}.
\end{align}
By induction,
$\text{ord}_p((f_p(p^\alpha(n-1))-f_p(p^\alpha)^{n-1})\geq 2\alpha$,
hence we only need to show that
\begin{align}
\text{ord}_p\left(f_p(p^\alpha n)-f_p(p^\alpha (n-1))
f_p(p^\alpha)\right)\geq 2\alpha.
\end{align}
By a straightforward computation,
\begin{align}
&f_p(p^\alpha n)-f_p(p^\alpha(n-1))f_p(p^\alpha)\\\nonumber
 &=f_p(p^\alpha(n-1))\left(\prod_{j=1,p\nmid j}^{p^\alpha}(p^\alpha(n-1)+j)-\prod_{j=1,p\nmid j}^{p^\alpha}
 j\right)\\\nonumber
 &=f_p(p^\alpha(n-1))f_p(p^\alpha)p^\alpha(n-1)\sum_{j=1,p\nmid
 j}^{p^{\alpha}}\frac{1}{j}.
\end{align}
For odd prime numbers $p$ and $\alpha\geq 1$ or for $p=2$,
$\alpha\geq 2$, then $p^{\alpha-1}(p-1)$ is even, thus
\begin{align}
\sum_{j=1,p\nmid
 j}^{p^{\alpha}}\frac{1}{j}=\sum_{j=1,p\nmid
 j}^{[p^\alpha/2]}\left(\frac{1}{j}+\frac{1}{p^{\alpha}-j}\right)=\sum_{j=1,p\nmid
 j}^{[p^\alpha/2]}\frac{p^{\alpha}}{j(p^{\alpha}-j)}.
\end{align}
Therefore, $\text{ord}_p\left(f_p(p^\alpha n)-f_p(p^\alpha (n-1))
f_p(p^\alpha)\right)\geq 2\alpha.$

As to the case $p=2$ and $\alpha=1$, note that $f_2(2n)=(2n-1)!!$,
then formula (\ref{formula-f2}) is easy to check by induction.
\end{proof}

In the following, suppose $r,s\geq 0$ are two given integers.

\begin{lemma} \label{lemma-2}
    For odd prime number $p$ such that $m=p^\alpha a, l_i=p^{\beta_i}
    b_i$,$k_j=p^{\gamma_j}c_j$,
     $p\nmid a, p\nmid b_i , p\nmid c_j$ for $1\leq i\leq r$, $1\leq j\leq s$ and $\alpha\geq 1, \beta_i \geq 0, \gamma_j\geq 0$, we have
    \begin{align} \label{formula-orderinequ1}
        &\text{ord}_p
        \left(\binom{m\tau+|\mathbf{l}|+|\mathbf{k}|-1}{m-1}\prod_{i=1}^{r}\binom{m}{l_i}\prod_{j=1}^s
\frac{m}{k_j}\binom{m+k_j-1}{k_j-1}\right.\\\nonumber
&\left.-\binom{\frac{m\tau+|\mathbf{l}|+|\mathbf{k}|}{p}-1}{\frac{m}{p}-1}\prod_{i=1}^{r}\binom{\frac{m}{p}}{\frac{l_i}{p}}\prod_{j=1}^s
\frac{m}{k_j}\binom{\frac{m+k_j}{p}-1}{\frac{k_j}{p}-1}\right)\geq
        2\alpha
    \end{align}
    where the second term is defined to be zero if one of $\beta_i$
    or $\gamma_j$ is zero.
\end{lemma}
\begin{proof}
\begin{align} \label{formula-equa1}
&\binom{m\tau+|\mathbf{l}|+|\mathbf{k}|-1}{m-1}\prod_{i=1}^{r}\binom{m}{l_i}\prod_{j=1}^s
\frac{m}{k_j}\binom{m+k_j-1}{k_j-1}\\\nonumber
&-\binom{\frac{m\tau+|\mathbf{l}|+|\mathbf{k}|}{p}-1}{\frac{m}{p}-1}\prod_{i=1}^{r}\binom{\frac{m}{p}}{\frac{l_i}{p}}\prod_{j=1}^s
\frac{m}{k_j}\binom{\frac{m+k_j}{p}-1}{\frac{k_j}{p}-1}\\\nonumber
&=
\binom{\frac{m\tau+|\mathbf{l}|+|\mathbf{k}|}{p}-1}{\frac{m}{p}-1}\prod_{i=1}^{r}\binom{\frac{m}{p}}{\frac{l_i}{p}}\prod_{j=1}^s
\frac{m}{k_j}\binom{\frac{m+k_j}{p}-1}{\frac{k_j}{p}-1}\\\nonumber
&\times\left(\frac{f_p(m\tau+|\mathbf{l}|+|\mathbf{k}|)}
{f_p(m)f_p(m(\tau-1)+|\mathbf{l}|+|\mathbf{k}|)}\prod_{i=1}^r\frac{f_p(m)}{f_p(l_i)f_p(m-l_i)}
\prod_{j=1}^s\frac{f_p(m)}{f_p(k_j)f_p(m-k_j)}-1\right)
\end{align}

By Lemma \ref{lemma-1}, we have
\begin{align} \label{formula-orderinequ2}
&\text{ord}_p\left(\frac{f_p(m\tau+|\mathbf{l}|+|\mathbf{k}|)}
{f_p(m)f_p(m(\tau-1)+|\mathbf{l}|+|\mathbf{k}|)}\prod_{i=1}^r\frac{f_p(m)}{f_p(l_i)f_p(m-l_i)}
\prod_{j=1}^s\frac{f_p(m)}{f_p(k_j)f_p(m-k_j)}-1\right)\\\nonumber
&\geq
2\min(\alpha,\beta_1,..,\beta_i,..,\beta_r,\gamma_1,..,\gamma_j,..,\gamma_s).
\end{align}
For brevity, we only compute
$\text{ord}_p\left(\frac{f_p(m)}{f_p(l_i)f_p(m-l_i)}-1\right)$, the
computation for (\ref{formula-orderinequ2}) is the same. Indeed, by
Lemma \ref{lemma-1}, we have
\begin{align}
&\text{ord}_p(\frac{f_p(m)}{f_p(l_i)f_p(m-l_i)}-1)\\\nonumber
&=\text{ord}_p\left(\frac{f_p(m)-f_p(l_i)f_p(m-l_i)}{f_p(l_i)f_p(m-l_i)}\right)\\\nonumber
&= \text{ord}_p\left(f_p(m)-f_p(l_i)f_p(m-l_i)\right)\\\nonumber
&\geq \min\left(\text{ord}_p
(f_p(m)-f_p(p^{\min(\alpha,\beta_i)})^{\frac{m}{\min(\alpha,\beta_i)}}),\right.\\\nonumber
&\left.\text{ord}_p(f_p(l_i)f_p(m-l_i)-f_p(p^{\min(\alpha,\beta_i)})^{\frac{m}{\min(\alpha,\beta_i)}})
\right)\\\nonumber & \geq 2 \min(\alpha,\beta_i).
\end{align}

In order to compute the orders of the other parts of the righthand
side of formula (\ref{formula-equa1}), we need to divide it into
different cases to discuss.

When $r=s=0$, formula (\ref{formula-orderinequ2}) implies that
formula (\ref{formula-orderinequ1}) holds.

When $r=1$ and $s=0$ ( or $r=0$ and $s=1$ ), if $\alpha>\beta$, then
\begin{align}
&\text{ord}_p\left(\binom{\frac{m\tau+l}{p}-1}{\frac{m}{p}-1}\binom{\frac{m}{p}}{\frac{l}{p}}\right)\\\nonumber
&=\text{ord}_p\left(\frac{m}{m\tau+l}\frac{m}{l}\binom{\frac{m\tau+l}{p}}{\frac{m}{p}}\binom{\frac{m}{p}-1}{\frac{l}{p}-1}\right)\\\nonumber
&=2(\alpha-\beta)
\end{align}
if $\alpha\leq \beta$, then
$\text{ord}_p\left(\binom{\frac{m\tau+l}{p}-1}{\frac{m}{p}-1}\binom{\frac{m}{p}}{\frac{l}{p}}\right)\geq
0$. We obtain the formula (\ref{formula-orderinequ1}) by using
formula (\ref{formula-orderinequ2}).

Now, we discuss the case when $r,s\geq 1$.

Case 1:
$\min(\alpha,\beta_1,..,\beta_i,..,\beta_r,\gamma_1,..,\gamma_j,..,\gamma_s)=\alpha$,
since
\begin{align}
\text{ord}_p\left(\binom{m\tau+|\mathbf{l}|+|\mathbf{k}|-1}{m-1}\prod_{i=1}^{r}\binom{m}{l_i}\prod_{j=1}^s
\frac{m}{k_j}\binom{m+k_j-1}{k_j-1}\right)\geq 0,
\end{align}
then together with formula (\ref{formula-orderinequ2}) gives
(\ref{formula-orderinequ1}).

Case 2:
$\min(\alpha,\beta_1,..,\beta_i,..,\beta_r,\gamma_1,..,\gamma_j,..,\gamma_s)\neq
\alpha$, and
$$\min(\alpha,\beta_1,..,\beta_i,..,\beta_r,\gamma_1,..,\gamma_j,..,\gamma_s)=\beta_1$$
without loss of generality.

Case 2a: all the other terms are bigger than $\beta_1$,
\begin{align}
&\text{ord}_p\left(\binom{m\tau+|\mathbf{l}|+|\mathbf{k}|-1}{m-1}\prod_{i=1}^{r}\binom{m}{l_i}\prod_{j=1}^s
\frac{m}{k_j}\binom{m+k_j-1}{k_j-1}\right)\\\nonumber &
=\text{ord}_p\left(\frac{m}{m\tau+|\mathbf{l}|+|\mathbf{k}|}\binom{m\tau+|\mathbf{l}|+|\mathbf{k}|}{m}\frac{m}{l_1}\binom{m-1}{l_1-1}\prod_{i=2}^{r}\binom{m}{l_i}\prod_{j=1}^s
\frac{m}{k_j}\binom{m+k_j-1}{k_j-1}\right)\\\nonumber &\geq
2(\alpha-\beta_1),
\end{align}
together with formula (\ref{formula-orderinequ2}) imply
(\ref{formula-orderinequ1}).

Case 2b: at least one of other $\beta_i$ or $\gamma_j$ equal to
$\beta_1$. Without loss of generality, suppose $\gamma_1=\beta_1$.
Then
\begin{align}
&\text{ord}_p\left(\binom{m\tau+|\mathbf{l}|+|\mathbf{k}|-1}{m-1}\prod_{i=1}^{r}\binom{m}{l_i}\prod_{j=1}^s
\frac{m}{k_j}\binom{m+k_j-1}{k_j-1}\right)\\\nonumber &
=\text{ord}_p\left(\binom{m\tau+|\mathbf{l}|+|\mathbf{k}|-1}{m-1}\frac{m}{l_1}\binom{m-1}{l_1-1}\prod_{i=2}^{r}\binom{m}{l_i}
\frac{m}{k_1}\binom{m+k_1-1}{k_1-1}\prod_{j=2}^s
\frac{m}{k_j}\binom{m+k_j-1}{k_j-1}\right)\\\nonumber & \geq
2(\alpha-\beta_1),
\end{align}
together with formula (\ref{formula-orderinequ2}) also imply
(\ref{formula-orderinequ1}).
\end{proof}

\begin{lemma} \label{lemma-3}
    For $m=2^\alpha a, l_i=2^{\beta_i}
    b_i$,$k_j=2^{\gamma_j}c_j$,
     $2\nmid a, 2\nmid b_i , 2\nmid c_j$ for $1\leq i\leq r$, $1\leq j\leq s$ and $\alpha\geq 1, \beta_i \geq 0, \gamma_j\geq 0$, we have
    \begin{align} \label{formula-orderinequ3}
        &\text{ord}_2
        \left((-1)^{m(\tau+1)+|\mathbf{l}|}\cdot\binom{m\tau+|\mathbf{l}|+|\mathbf{k}|-1}{m-1}\prod_{i=1}^{r}\binom{m}{l_i}\prod_{j=1}^s
\frac{m}{k_j}\binom{m+k_j-1}{k_j-1}\right.\\\nonumber
&\left.-(-1)^{\frac{m(\tau+1)+|\mathbf{l}|}{2}}\binom{\frac{m\tau+|\mathbf{l}|+|\mathbf{k}|}{2}-1}{\frac{m}{2}-1}\prod_{i=1}^{r}\binom{\frac{m}{2}}{\frac{l_i}{2}}\prod_{j=1}^s
\frac{m}{k_j}\binom{\frac{m+k_j}{2}-1}{\frac{k_j}{2}-1}\right)\geq
        2\alpha
    \end{align}
    where the second term is defined to be zero if one of $\beta_i$
    or $\gamma_j$ is zero.
\end{lemma}
\begin{proof}
Case 1: all the $\alpha,\beta_i,\gamma_j\geq 2$ , then
$(-1)^{m(\tau+1)+|\mathbf{l}|}=(-1)^{\frac{m(\tau+1)+|\mathbf{l}|}{2}}=1$,
in this case the proof is the same as in Lemma \ref{lemma-2}.

Case 2: only one of $\beta_i$ (or $\gamma_j$) is equal to zero, then
\begin{align}
&\text{ord}_2\binom{m(\tau+1)+|\mathbf{l}|+|\mathbf{k}|-1}{m-1}\\\nonumber
&=\text{ord}_2\left(\frac{m}{m(\tau+1)+|\mathbf{l}|+|\mathbf{k}|}
\binom{m(\tau+1)+|\mathbf{l}|+|\mathbf{k}|}{m}\right)\geq \alpha
\end{align}
together with
\begin{align} \label{formula-betai}
&\text{ord}_2\binom{m}{l_i}=\text{ord}_2\left(\frac{m}{l_i}\binom{m-1}{l_i-1}\right)\geq
\alpha
\end{align}
imply formula (\ref{formula-orderinequ3}).

Case 3: at least two $\beta_i$ or $\gamma_j$ (suppose they are
$\beta_i$ and $\gamma_j$) are equal to zero, then
\begin{align}
&\text{ord}_2\binom{m}{k_j}=\text{ord}_2\left(\frac{m}{k_j}\binom{m-1}{k_j-1}\right)\geq
\alpha.
\end{align}
together with formula (\ref{formula-betai}) imply  formula
(\ref{formula-orderinequ3}).

For the remain cases, we compute similarly as in
(\ref{formula-equa1})

\begin{align} \label{formula-equa2}
&(-1)^{m(\tau+1)+|\mathbf{l}|}\binom{m\tau+|\mathbf{l}|+|\mathbf{k}|-1}{m-1}\prod_{i=1}^{r}\binom{m}{l_i}\prod_{j=1}^s
\frac{m}{k_j}\binom{m+k_j-1}{k_j-1}\\\nonumber
&-(-1)^{\frac{m(\tau+1)+|\mathbf{l}|}{2}}\binom{\frac{m\tau+|\mathbf{l}|+|\mathbf{k}|}{2}-1}{\frac{m}{2}-1}
\prod_{i=1}^{r}\binom{\frac{m}{2}}{\frac{l_i}{2}}\prod_{j=1}^s
\frac{m}{k_j}\binom{\frac{m+k_j}{2}-1}{\frac{k_j}{2}-1}\\\nonumber
&=
(-1)^{m(\tau+1)+|\mathbf{l}|}\binom{\frac{m\tau+|\mathbf{l}|+|\mathbf{k}|}{2}-1}{\frac{m}{2}-1}
\prod_{i=1}^{r}\binom{\frac{m}{2}}{\frac{l_i}{2}}\prod_{j=1}^s
\frac{m}{k_j}\binom{\frac{m+k_j}{2}-1}{\frac{k_j}{2}-1}\\\nonumber
&\times\left(\frac{f_2(m\tau+|\mathbf{l}|+|\mathbf{k}|)}
{f_2(m)f_2(m(\tau-1)+|\mathbf{l}|+|\mathbf{k}|)}\prod_{i=1}^r\frac{f_2(m)}{f_2(l_i)f_2(m-l_i)}
\prod_{j=1}^s\frac{f_2(m)}{f_2(k_j)f_2(m-k_j)}-(-1)^{\frac{m(\tau+1)+|\mathbf{l}|}{2}}\right).
\end{align}

For the case $\alpha=1$, it remains to show that
\begin{align} \label{formula-inequ3}
\text{ord}_{2}&\left(\frac{f_2(m\tau+|\mathbf{l}|+|\mathbf{k}|)}
{f_2(m)f_2(m(\tau-1)+|\mathbf{l}|+|\mathbf{k}|)}\prod_{i=1}^r\frac{f_2(m)}{f_2(l_i)f_2(m-l_i)}\right.\\\nonumber\times
&\left.
\prod_{j=1}^s\frac{f_2(m)}{f_2(k_j)f_2(m-k_j)}-(-1)^{\frac{m(\tau+1)+|\mathbf{l}|}{2}}\right)\geq
2.
\end{align}

For the case $\alpha\geq 2$,  if only one of $\beta_i$ (or
$\gamma_j$) is equal to 1, then
\begin{align}
&\text{ord}_2\left(\binom{\frac{m\tau+|\mathbf{l}|+|\mathbf{k}|}{2}-1}{\frac{m}{2}-1}
\binom{\frac{m}{2}}{\frac{l_i}{2}}\right)\\\nonumber
&=\text{ord}_2\left(\frac{m}{m\tau+|\mathbf{l}|+|\mathbf{k}|}\binom{\frac{m\tau+|\mathbf{l}|+|\mathbf{k}|}{2}-1}{\frac{m}{2}-1}
\frac{m}{l_i}\binom{\frac{m}{2}-1}{\frac{l_i}{2}-1}\right)\geq
2(\alpha-1).
\end{align}

If at least two $\beta_i$ or $\gamma_j$ (suppose they are $\beta_i$
and $\gamma_j$) are equal to 1, then
\begin{align}
&\text{ord}_2\left(\binom{\frac{m}{2}}{\frac{l_i}{2}}\frac{m}{k_j}\right)=\text{ord}_2\left(\frac{m}{k_j}
\frac{m}{l_i}\binom{\frac{m}{2}-1}{\frac{l_i}{2}-1}\right)\geq
2(\alpha-1).
\end{align}

Therefore, it also remains to show the inequality
$(\ref{formula-inequ3})$ which can be obtained by applying the Lemma
\ref{lemma-1}. We leave the details to the reader.

\end{proof}

Now, we can finish the proof of Theorem \ref{theorem-introduction}.
\begin{proof}
\begin{align}
n_{m,\mathbf{l},\mathbf{k}}(\tau)&=\sum_{d|\text{gcd}(m,\mathbf{l},\mathbf{k})}\frac{\mu(d)}{d^2}c_{m/d,\mathbf{l}/d,\mathbf{k}/d}(\tau)\\\nonumber
&=\frac{1}{m^2}\sum_{d|\text{gcd}(m,\mathbf{l},\mathbf{k})}\mu(d)\cdot(-1)^{\frac{m(\tau+1)+|\mathbf{l}|}{d}}
\binom{\frac{m\tau+|\mathbf{l}|+|\mathbf{k}|}{d}-1}{\frac{m}{d}-1}\\\nonumber
&\times\prod_{j=1}^{r}\binom{\frac{m}{d}}{\frac{l_j}{d}}\prod_{j=1}^s
\frac{m}{m+k_j}\binom{\frac{m+k_j}{d}}{\frac{k_j}{d}}.
\end{align}

Suppose we have the prime factorization
$m=p_1^{\alpha_1}p_2^{\alpha_2}\cdots p_n^{\alpha_n}$, we only need
to show that the summation term is divisible by $p_t^{2\alpha_t}$
for every $1\leq t\leq n$.

Given any such $p_t$, if $p_t\nmid
\text{gcd}(m,\mathbf{l},\mathbf{k})$, then by Lemma and , every
terms in the above summation is divisible by $p_t^{2\alpha_t}$.

if $p_t|\text{gcd}(m,\mathbf{l},\mathbf{k})$, then
\begin{align}
&\sum_{d|\text{gcd}(m,\mathbf{l},\mathbf{k})}\mu(d)\cdot(-1)^{\frac{m(\tau+1)+|\mathbf{l}|}{d}}
\binom{\frac{m\tau+|\mathbf{l}|+|\mathbf{k}|}{d}-1}{\frac{m}{d}-1}\prod_{j=1}^{r}\binom{\frac{m}{d}}{\frac{l_j}{d}}\prod_{j=1}^s
\frac{m}{m+k_j}\binom{\frac{m+k_j}{d}}{\frac{k_j}{d}}\\\nonumber
&=\pm
\sum_{d_t|\frac{\text{gcd}(m,\mathbf{l},\mathbf{k})}{p_t}}\left((-1)^{\frac{m(\tau+1)+|\mathbf{l}|}{d_t}}
\binom{\frac{m\tau+|\mathbf{l}|+|\mathbf{k}|}{d_t}-1}{\frac{m}{d_t}-1}\prod_{j=1}^{r}\binom{\frac{m}{d_t}}{\frac{l_j}{d_t}}\prod_{j=1}^s
\frac{m}{m+k_j}\binom{\frac{m+k_j}{d_t}}{\frac{k_j}{d_t}}\right.\\\nonumber
&\left.-(-1)^{\frac{m(\tau+1)+|\mathbf{l}|}{d_tp_t}}
\binom{\frac{m\tau+|\mathbf{l}|+|\mathbf{k}|}{d_tp_t}-1}{\frac{m}{d_tp_t}-1}\prod_{j=1}^{r}\binom{\frac{m}{d_tp_t}}{\frac{l_j}{d_tp_t}}\prod_{j=1}^s
\frac{m}{m+k_j}\binom{\frac{m+k_j}{d_tp_t}}{\frac{k_j}{d_tp_t}}
\right).
\end{align}

By Lemma \ref{lemma-2} and Lemma \ref{lemma-3}, the above terms is
divisible by $p_t^{2\alpha_t}$.
\end{proof}

\section{Quantum $2$-functions} \label{Section-quantum2functions}

\subsection{$2$-functions}
Motivated by the multiple covering formulas (\ref{formula-genus0})
and (\ref{formula-disk}), Schwarz, Vologodsky and Walcher
\cite{SVW1}, introduced the notion of {\em $s$-function} as integral
linear combinations of poly-logarithms. Here we review the
definition  of $2$-function.
\begin{definition} \label{definition-2function}
Given $t$ variables $z_1,...z_t$, we call a  formal power series
\begin{align}
W(z_1,...,z_t)=\sum_{d_1,..,d_t\geq 1}m_{d_1,...,d_t}z_1^{d_1}\cdots
z_t^{d_t}\in \mathbb{Q}[[z_1,...,z_t]]
\end{align}
with rational coefficients $m_{d_1,...,d_t}$ a $2$-function if it
can be written as an integral linear combination of di-lograrithms
\begin{align}
W(z_1,...,z_t)=\sum_{d_1,...,d_t\geq
1}n_{d_1,...,d_t}\text{Li}_2(z_1^{d_1}\cdots z_t^{d_t}).
\end{align}

\end{definition}

\begin{lemma} \label{Lemma-2function}
$W(z_1,...,z_t)\in \mathbb{Q}[[z_1,...,z_t]]$ is an $2$-function if
and only if
$$m_{d_1,...,d_t}-\frac{1}{p^2}m_{\frac{d_1}{p},\frac{d_2}{p},...,\frac{d_t}{p}}$$ is
$p$-integral for all $p,d_1,...,d_t$, where
$m_{\frac{d_1}{p},\frac{d_2}{p},...,\frac{d_t}{p}}=0$ if $p\nmid
\text{gcd}(d_1,...,d_t)$.
\end{lemma}
\begin{proof}
The proof is essentially given in \cite{KSV,SVW1}. By using the
formula
\begin{align}
W(z_1,..,z_t)&=\sum_{d_1,..,d_t\geq 1} m_{d_1,...,d_t}
z_1^{d_1}\cdots z_t^{d_t}\\\nonumber &=\sum_{d_1,d_2,..,d_t,k\geq
1}\frac{n_{d_1,d_2,..,d_t}}{k^2}(z_1^{d_1}\cdots z_t^{d_t})^k
\end{align}
we obtain
\begin{align} \label{formula-mexpression}
m_{d_1,...,d_t}=\sum_{k|\text{gcd}(d_1,...,d_t)}\frac{1}{k^2}n_{\frac{d_1}{k},..,\frac{d_t}{k}}.
\end{align}
Applying the M\"obius inversion formula, we find
\begin{align} \label{formula-nexpression}
n_{d_1,...,d_t}=\sum_{k|\text{gcd}(d_1,...,d_t)}\frac{\mu(k)}{k^2}m_{\frac{d_1}{k},...,\frac{d_t}{k}}.
\end{align}
Indeed, by the definition of M\"obius function $\mu(k)=0$ if $k$ is
not squarefree, $\mu(k)=(-1)^r$ if $k=p_1\cdots p_r$ is the product
of $r$ distinct primes. It follows that
\begin{align} \label{formula-mobiussum}
\sum_{k| d}\mu(k)=\delta_{1d}.
\end{align}

Applying formula (\ref{formula-mobiussum}), we obtain
\begin{align}
\sum_{k|\text{gcd}(d_1,...,d_t)}\frac{\mu(k)}{k^2}m_{\frac{d_1}{k},...,\frac{d_t}{k}}
&=\sum_{k|\text{gcd}(d_1,...,d_t)}\frac{\mu(k)}{k^2}\sum_{l|\text{gcd}(\frac{d_1}{k},...,\frac{d_t}{k})}\frac{1}{l^2}n_{\frac{d_1}{kl},...,\frac{d_t}{kl}}\\\nonumber
&=\sum_{p|\text{gcd}(d_1,...,d_t)}\sum_{k|p}\frac{\mu(k)}{p^2}n_{\frac{d_1}{p},...,\frac{d_t}{p}}\\\nonumber
&=n_{d_1,..,d_t}
\end{align}
which is just the formula (\ref{formula-nexpression}).

$\Rightarrow$: by the formula (\ref{formula-mexpression}),
\begin{align}
m_{d_1,...,d_t}-\frac{1}{p^2}m_{\frac{d_1}{p},...,\frac{d_t}{p}}&=\sum_{k|\text{gcd}(d_1,..,d_t)}\frac{1}{k^2}
n_{\frac{d_1}{k},...,\frac{d_t}{k}}-\frac{1}{p^2}\sum_{l|\frac{\text{gcd}(d_1,..,d_t)}{p}}\frac{1}{l^2}n_{\frac{d_1}{pl},..,\frac{d_t}{pl}}\\\nonumber
&=\sum_{\substack{k|\text{gcd}(d_1,..,d_t)\\ p\nmid k
}}\frac{1}{k^2}n_{\frac{d_1}{k},...,\frac{d_t}{k}}
\end{align}
Note the sum is restricted to those $k$ has no prime factor $p$, and
therefore the righthand side is $p$-integral if for any
$n_{d_1,..,d_t}\in \mathbb{Z}$.

$\Leftarrow$: since $\mu(k)=0$ if $k$ is divisible by $p^2$, and
$\mu(pk)=-\mu(k)$ if $p\nmid k$,  by formula
(\ref{formula-nexpression}), we get
\begin{align}
n_{d_1,...,d_t}&=\sum_{k|\text{gcd}(d_1,...,d_t)}\frac{\mu(k)}{k^2}m_{\frac{d_1}{k},...,\frac{d_t}{k}}\\\nonumber
&=\sum_{\substack{k|\text{gcd}(d_1,...,d_t)\\
p\nmid
k}}\frac{\mu(k)}{k^2}\left(m_{\frac{d_1}{k},...,\frac{d_t}{k}}-\frac{1}{p^2}m_{\frac{d_1}{pk},...,\frac{d_t}{pk}}\right)
\end{align}
with the same understanding that
$m_{\frac{d_1}{pk},...,\frac{d_t}{pk}}=0$ if $p\nmid d$. We see that
if
$m_{d_1,...,d_t}-\frac{1}{p^2}m_{\frac{d_1}{p},...,\frac{d_t}{p}}$
are $p$-integral for all $p,d_1,...,d_t$,  then $n_{d_1,...,d_t}$
are $p$-integral for any $p$, hence integral.
\end{proof}

Indeed, the proof Theorem \ref{theorem-introduction} implies that
\begin{theorem}
The disc counting formula for generalized conifold given by formula
(\ref{formula-disk-generalizedconifold})
\begin{align}
F_{\text{disk}}^{(\widehat{X},D_\tau)}=\sum_{m,\mathbf{l},\mathbf{k}}c_{m,\mathbf{l},\mathbf{k}}(\tau)x^m
a_1^{l_1}\cdots a_r^{l_r}A_1^{k_1}\cdots A_s^{l_s}
\end{align}
is a 2-function.
\end{theorem}
\begin{proof}
By Lemma \ref{Lemma-2function}, we only need to show that
\begin{align}
c_{m,\mathbf{l},\mathbf{k}}(\tau)-\frac{1}{p^2}c_{\frac{m}{p},\frac{\mathbf{l}}{p},\frac{\mathbf{k}}{p}}(\tau)
\end{align}
is $p$-integral for all prime $p$ and positive integers
$m,\mathbf{l},\mathbf{k}$.

Note that this statement is an easy consequence of Lemma and  .
Indeed, by formula (\ref{formula-cmlk}), we have
\begin{align}
&c_{m,\mathbf{l},\mathbf{k}}(\tau)-\frac{1}{p^2}c_{\frac{m}{p},\frac{\mathbf{l}}{p},\frac{\mathbf{k}}{p}}(\tau)\\\nonumber
&=\frac{1}{m^2}\left((-1)^{m(\tau+1)+|\mathbf{l}|}
\binom{m\tau+|\mathbf{l}|+|\mathbf{k}|-1}{m-1}\prod_{j=1}^{r}\binom{m}{l_j}\prod_{j=1}^s
\frac{m}{m+k_j}\binom{m+k_j}{k_j}\right.\\\nonumber
&-\left.(-1)^{\frac{m(\tau+1)+|\mathbf{l}|}{p}}
\binom{\frac{m\tau+|\mathbf{l}|+|\mathbf{k}|}{p}-1}{\frac{m}{p}-1}\prod_{j=1}^{r}\binom{m}{l_j}\prod_{j=1}^s
\frac{m}{m+k_j}\binom{\frac{m+k_j}{p}}{\frac{k_j}{p}}\right).
\end{align}
Then Lemma \ref{lemma-2} and Lemma \ref{lemma-3} implies that
$c_{m,\mathbf{l},\mathbf{k}}(\tau)-\frac{1}{p^2}c_{\frac{m}{p},\frac{\mathbf{l}}{p},\frac{\mathbf{k}}{p}}(\tau)$
is $p$-integral for all $p,d_1,..,d_t$.
\end{proof}

\subsection{Quantum $2$-functions}
Motivated by Ooguri-Vafa's work \cite{OV} which generalized the disc
counting formula  (\ref{formula-disk}) to  the higher genus case, we
introduce the notion of quantum $2$-function, that means there
exists a deformation parameter $\lambda$, such that when
$\lambda\rightarrow 0$, the quantum $2$-function reduced to the
$2$-function in the sense of Schwarz-Vologodsky-Walcher.

For convenience, we introduce some notations first. We set
$\mathbf{z}=(z_1,...,z_r)$ for $r$ variables $z_1,..z_r$, and
$\mathbf{d}=(d_1,...,d_r)$ for $r$ nonnegative integers
$d_1,...,d_r$, in particular  $\mathbf{0}=(0,...,0)$. Then we denote
$\mathbf{z}^{\mathbf{d}}=\prod_{i=1}^rz_i^{d_i}$.

\begin{definition}  We call a  formal
power series
\begin{align} \label{formula-F}
F(\lambda,\mathbf{z},x)=\sum_{g\geq 0,m\geq 1}\sum_{\mathbf{d}>
 \mathbf{0}}\lambda^{2g}K_{g,\mathbf{d},m}\mathbf{z}^{\mathbf{d}}x^m\in
\mathbb{Q}[[\lambda^2,z_1,...,z_r,x]]
\end{align}
with rational coefficients $K_{g,\mathbf{d},m}$ a {\em quantum
$2$-function} if it can be written in the following form
\begin{align} \label{formula-quantum2function}
F(\lambda,\mathbf{z},x)&=\sum_{g\geq 0,m\geq 1}\sum_{\mathbf{d}>
\mathbf{0}}\sum_{k\geq
1}n_{g,\mathbf{d},m}\frac{2\lambda}{km}\sin\left(\frac{km\lambda}{2}\right)
\left(2\sin
\frac{k\lambda}{2}\right)^{2g-2}\mathbf{z}^{\mathbf{kd}}x^{km}
\end{align}
with $n_{g,\mathbf{d},m}\in \mathbb{Z}$.
\end{definition}
It is clear that when the parameter $\lambda=0$, $F(0,\mathbf{z},x)$
is just the $2$-function in the sense of Definition
\ref{definition-2function}.

For convenience, let
$
\hat{F}(\lambda,\mathbf{z},x)=\sqrt{-1}\lambda^{-1}F(\lambda,\mathbf{z},x).
$
We set $q=e^{\sqrt{-1}\lambda}$ and
$\{m\}=q^{\frac{m}{2}}-q^{-\frac{m}{2}}$. Let
$\hat{n}_{g,\mathbf{d},m}=(-1)^{g-1}n_{g,\mathbf{d},m}$, and we
introduce the function
\begin{align}
\hat{f}_m(q,\mathbf{z})=\sum_{g\geq 0}\sum_{\mathbf{d}>
\mathbf{0}}\hat{n}_{g,\mathbf{d},m}(q^{\frac{1}{2}}-q^{-\frac{1}{2}})^{2g-2}\mathbf{z}^{\mathbf{d}}.
\end{align}
For $k\in \mathbb{Z}_+$, we define the $k$-th Adams operator
$\Psi_k$ as the $\mathbb{Q}$-algebra map on
$\mathbb{Q}(q^{\frac{1}{2}})[[x,\mathbf{z}]]$ by
\begin{align}
\Psi_{k}(g(x,q,\mathbf{z}))= g(x^k,q^k,\mathbf{z}^k).
\end{align}

Then, the formula (\ref{formula-quantum2function}) can be rewritten
as
\begin{align}
\hat{F}(q,\mathbf{z},x)=\sum_{k\geq
1}\frac{1}{k}\Psi_k\left(\sum_{m\geq
1}\frac{\{m\}}{m}\hat{f}_{m}(q,\mathbf{z})x^m\right).
\end{align}

By M\"obius inversion formula, we obtain
\begin{align}
\sum_{m\geq 1}\hat{f}_m(q,\mathbf{z})x^m=\sum_{k\geq
1}\frac{\mu(k)}{k}\Psi_k\left(\hat{F}(q,\mathbf{z},x)\right)=\sum_{k\geq
1}\frac{\mu(k)}{k}\hat{F}(q^k,\mathbf{z}^k,x^k).
\end{align}
Therefore, for $m\geq 1$, we have
\begin{align} \label{formula-f}
\hat{f}_m(q,\mathbf{z})=\frac{\sqrt{-1}}{\{m\}}\sum_{k|m}\frac{\mu(k)}{k^{2-2g}}\sum_{g\geq
0,\mathbf{d}>0}\lambda^{2g-1}K_{g,\mathbf{d},\frac{m}{k}}\mathbf{z}^{k\mathbf{d}}.
\end{align}

In conclusion, the function $F(\lambda,\mathbf{z},x)$ given by
formula (\ref{formula-F}) is a quantum $2$-function if and only the
function $\hat{f}_m(q,\mathbf{z})$ given by formula
(\ref{formula-f}) belongs to the ring
$z^{-2}\mathbb{Z}[z^{2},\mathbf{z}]$, where
$z=q^{\frac{1}{2}}-q^{-\frac{1}{2}}$.

However, in general, it is difficult to show the above statement for
a function given by formula (\ref{formula-F}). Based on the works
\cite{OV,LM1,LM2,LMV}, it is expected that the generating functions
of certain type open Gromov-Witten invariants in topological string
theory provide many examples of quantum $2$-functions.

Let us study the basic model $(\hat{X},D_\tau)$ with $\hat{X}$ the
resolved conifold and $D_\tau$ the Aganagic-Vafa A-brane which is
the dual of framed unknot $U_\tau$ with framing $\tau$ in
Chern-Simons theory. We consider the generating function
\begin{align} \label{formula-F-conifold}
F^{(\hat{X},D_\tau)}(\lambda,a,x)=\sum_{g\geq 0,m\geq 1}\sum_{d>
 0}\lambda^{2g}K_{g,d,m}^{(\hat{X},D_\tau)}a^{d}x^m
\end{align}
where $K_{g,d,m}^{(\hat{X},D_\tau)}$ are the one-hole genus $g$ open
Gromov-Witten invariants with degree $d$ and writhe number $m$,
whose mathematical definition was given in \cite{KL}.

According to the Mari\~no-Vafa's formula proposed in \cite{MV}, and
proved by \cite{LLZ,Zhou}, one can show that the corresponding
formula (\ref{formula-f}) in this case, denoted by
\begin{align} \label{formula-f-conifold}
\hat{f}_m^{(\hat{X},D_\tau)}(q,a)
\end{align}
can be given as follow:

Let $n\in \mathbb{Z}$ and $\mu,\nu$ denote the partitions. We
introduce the following notations
\begin{align*}
\{n\}_x=x^{\frac{n}{2}}-x^{-\frac{n}{2}}, \
\{\mu\}_{x}=\prod_{i=1}^{l(\mu)}\{\mu_i\}_x.
\end{align*}
In particular, let $\{n\}=\{n\}_q$ and $\{\mu\}=\{\mu\}_q$.

Let
\begin{align*}
\mathcal{Z}_m(q,a)=(-1)^{m\tau}\sum_{|\nu|=m}\frac{1}{\mathfrak{z}_\nu}\frac{\{m\nu\tau\}}{\{m\}\{m\tau\}}\frac{\{\nu\}_a}{\{\nu\}}
\end{align*}
where $\mathfrak{z}_{\nu}=|Aut(\nu)|\prod_{i=1}^{l(\nu)}\nu_i$ and
$\{m\}$ denotes the quantum integer, see Section 2 in \cite{LZ} for
 these notations. Then we
have the following formula for the expression
(\ref{formula-f-conifold})
\begin{align}
\hat{f}_m^{(\hat{X},D_\tau)}(q,a)=\sum_{d|m}\mu(d)\mathcal{Z}_{m/d}(q^d,a^d).
\end{align}
In \cite{LZ},  we have proved that, for any $m\geq 1$,
\begin{align}
\hat{f}_m^{(\hat{X},D_\tau)}(q,a)\in z^{-2}\mathbb{Z}[z^2,a^{\pm
1}],
\end{align}
where $z=q^{\frac{1}{2}}-q^{-\frac{1}{2}}$. Therefore, we have
\begin{theorem}
The function $F^{(\hat{X},D_\tau)}(\lambda,a,x)$ given by formula
(\ref{formula-F-conifold}) is a quantum $2$-function.
\end{theorem}

\subsection{Quantization and framing transformation}
We have shown that the $2$-function  can be viewed as the classical
limit of quantum $2$-function, that's why we use the terminology
``quantum'' here. Now we study the converse question, given a
2-function, how to construct its quantization?

Let $\mathbf{T}$  and  $\mathbf{qT}$ denote the set of 2-functions
and the set of quantum 2-functions respectively, we need to
construct a quantized operator $\mathbf{Q}$ from $\mathbf{T}$ to
$\mathbf{qT}$.

Motivated by the method of topological recursion introduced in
\cite{EO1} and its applications in topological string theory
\cite{BKMP,EO2,FLZ}, we briefly describe a natural way to construct
this operator $\mathbf{Q}$.

First, comparing to the relationship between the superpotential and
mirror curve in topological string, one can construct a spectral
curve $\mathcal{C}_W$ for a given 2-function $W\in \mathbf{T}$.
Next, we apply the method of topological recursion $\cite{EO1}$ to
this spectral curve $\mathcal{C}_W$, and we will obtain a series of
symplectic invariants $\{F_{g,n}(\mathcal{C}_W)\}$. Finally, we
collect all the $n=1$ terms $F_{g,1}(\mathcal{C}_W)$ to construct a
generating function $F(\mathcal{C}_W)$. Then we hope that
$F(\mathcal{C}_W)$ is the expected quantum $2$-function, in other
words, we have
\begin{conjecture} \label{conjecture}
$F(\mathcal{C}_W)$ is a quantum $2$-function for any $W\in
\mathbf{T}$.
\end{conjecture}
The Conjecture \ref{conjecture} allows us to introduce a formal
operator $\mathbf{Q}: \mathbf{T}\rightarrow \mathbf{qT}$ by defining
$\mathbf{Q}(W)=F(\mathcal{C}_W)$ for any $2$-function $W\in
\mathbf{T}$.

On the other hand side, motivated by the notion of the framing
introduced in \cite{AV,AKV} which describes the ambiguity in toric
computations,  Schwarz-Vologodsky-Walcher considered the framing
transformation on $2$-function. For any $\tau\in \mathbb{Z}$, there
is a framing transformation operator $\mathbf{f}^{\tau}$. The main
result stated in \cite{SVW1} is that, for any $2$-function $W$,
$\mathbf{f}^{\tau}(W)$ is also a 2-function.

Hence one can lift the framing transformation operator
$\mathbf{f}^\tau$ from $\mathbf{T}$ to $\mathbf{Q}(\mathbf{T})$,
denoting the resulting operator by $\hat{\mathbf{f}}^{\tau}$, then
\begin{align}
\hat{\mathbf{f}}^{\tau}(\mathbf{Q}(W))=\mathbf{Q}(\mathbf{f}^\tau(W)).
\end{align}

Therefore, by Conjecture \ref{conjecture}, we obtain a lot of
quantum $2$-functions by quantization and framing transformations.
We leave the further study of the quantized operator $\mathbf{Q}$ to
a separated paper.

\section{Discussions and further questions}
\label{section-discussions} In this final section, we give some
related questions which deserve to be studied further.

 1. Finding more examples of $2$-functions and quantum 2-functions.
The existed examples of $2$-functions given in \cite{SVW1} are the
superpotentials or disc counting formulas in open topological string
theory. Motivated by the large $N$ duality of Chern-Simons and
topological string theory, the Chern-Simons partition function of a
knot which is a generating function of colored HOMFLYPT invariants
of the knot \cite{Zhu1}, carries the natural integrality structure
inherited from topological string theory. This statement is referred
as to be the Labastida-Mari\~no-Ooguri-Vafa (LMOV) conjecture in
\cite{LP1,CLPZ}. Therefore,  one can define the (quantum) 2-function
for any knot/link via the LMOV conjecture. If we consider the framed
knot $\mathcal{K}_\tau$ with an integer framing $\tau$, the
corresponding framed LMOV conjecture was studied in \cite{CLPZ}. It
is expected that the quantum 2-function of the framed knot
$\mathcal{K}_\tau$ can be written as
$\mathbf{Q}(\mathbf{f}^{\tau}(W))$, where $W$ is the 2-function of
the knot $\mathcal{K}_0$ with zero framing.

2. Studying the open topological string model beyond the
Aganagic-Vafa A-brane.  For example, Zaslow et al's works
\cite{TZ,Zas} proposed the wavefunction for some Lagrangian brane
which are asymptotic to Legendrian surface of genus $g$, they
conjectured the wavefunction encodes all-genus open Gromov-Witten
invariants. Therefore, one can derive a quantum 2-function from this
wavefunction. The basic number theory method used in Section
\ref{section-proof} can be applied to prove the integrality of some
formulas appearing in \cite{TZ,Zas}.

3. In \cite{SVW2}, the concept of $2$-function was generalized to
the situation of algebraic number field by replacing the rational
number field $\mathbb{Q}$ with algebraic number field $K$ in its
definition,  this generalization was motivated by the work in
topological string \cite{Wal3}. So it is also interesting to
consider the quantum 2-function in the situation of algebraic number
field.

\end{document}